\documentclass[12pt]{article}
\pdfoutput=1
\usepackage{amsmath,amssymb,amsthm,bm,graphicx,float,array,multirow,multicol,rotfloat,caption,subcaption,hyperref,cleveref,enumerate,geometry,mathdots,adjustbox,booktabs,parskip,mathtools,tikz,pdflscape,csquotes,lscape,rotating,empheq,braket}
\usepackage[para]{threeparttable}
\usepackage[all]{xy}
\usepackage[normalem]{ulem}
\usepackage[numbers,sort&compress]{natbib}
\numberwithin{equation}{section}
\numberwithin{figure}{section}
\allowdisplaybreaks
\makeatletter
\usetikzlibrary{arrows, positioning, decorations.pathmorphing, decorations.markings, decorations.pathreplacing, decorations.markings, matrix, patterns}
\hypersetup{colorlinks=true}
\hypersetup{linkcolor=black}
\hypersetup{citecolor=black}
\hypersetup{urlcolor=black}
\makeatletter

\theoremstyle{plain}
\newtheorem*{thm*}{Theorem}
\newtheorem{thm}{Theorem}[section]

\newtheorem{lem}[thm]{Lemma}
\newtheorem{prop}[thm]{Proposition}

\theoremstyle{definition}
\newtheorem{defn}[thm]{Definition}
\newtheorem*{defn*}{Definition}

\crefname{lemma}{lemma}{lemmas}
\Crefname{lemma}{Lemma}{Lemmas}
\crefname{thm}{theorem}{theorems}
\Crefname{thm}{Theorem}{Theorems}
\crefname{defn}{definition}{definitions}
\Crefname{defn}{Definition}{Definitions}

\DeclarePairedDelimiterX{\abs}[1]{\lvert}{\rvert}{\ifblank{#1}{{}\cdot{}}{#1}}
\makeatother
\newcommand{\calo}{\mathcal{O}}

\newcommand{\cals}{\mathcal{S}}
\newcommand{\calh}{\mathcal{H}}

\newtheorem*{thm:main}{Theorem \ref{thm:main}}
\newtheorem*{thm:prop}{Proposition \ref{thm:prop}}

\begin{document}

\begin{titlepage}
\vspace*{-3cm} 
\begin{flushright}
{\tt CALT-TH-2021-042}\\
\end{flushright}
\begin{center}
\vspace{2.2cm}
{\LARGE\bfseries Nonperturbative gravity corrections\\ to bulk reconstruction \\}
\vspace{1.8cm}
{\large
Elliott Gesteau$^1$ and Monica Jinwoo Kang$^{1,2,3}$\\}
\vspace{.6cm}
{ $^1$ Division of Physics, Mathematics, and Astronomy, California Institute of Technology}\par\vspace{-.3cm}
{Pasadena, CA 91125, U.S.A.}\par
\vspace{.2cm}
{ $^2$ Walter Burke Institute for Theoretical Physics, California Institute of Technology}\par\vspace{-.3cm}
{Pasadena, CA 91125, U.S.A.}\par
\vspace{.2cm}
{ $^3$ Department of Physics, Korea Advanced Institute of Science and Technology}\par
{Daejeon 34141, Republic of Korea}\par
\vspace{.4cm}

\scalebox{.95}{\tt  egesteau@caltech.edu, monica@caltech.edu}\par
\vspace{1cm}
{\bf{Abstract}}\\
\end{center}
We introduce a new algebraic framework for understanding nonperturbative gravitational aspects of bulk reconstruction with a finite or infinite-dimensional boundary Hilbert space.
We use relative entropy equivalence between bulk and boundary with an inclusion of nonperturbative gravitational errors, which give rise to approximate recovery. We utilize the privacy/correctability correspondence to prove that the reconstruction wedge, the intersection of all entanglement wedges in pure and mixed states, manifestly satisfies bulk reconstruction. We explicitly demonstrate that local operators in the reconstruction wedge of a given boundary region can be recovered in a state-independent way for arbitrarily large code subspaces, up to nonperturbative errors in $G_N$.
We further discuss state-dependent recovery beyond the reconstruction wedge and the use of the twirled Petz map as a universal recovery channel. We discuss our setup in the context of quantum islands and the information paradox.
\\
\vfill 
\end{titlepage}

\tableofcontents
\newpage
\section{Introduction}

Understanding the holographic principle and the theory of quantum gravity has been a constant focus and desire in modern theoretical physics, and much progress has been achieved within the last decade utilizing the framework of quantum error correction. In the context of the AdS/CFT correspondence, the semi-classical limit can be interpreted as a quantum error correcting code and this perspective provided new tools to understand the emergence of spacetime \cite{Dong:2016eik,Harlow:2018fse,Gesteau:2020wrk} and the Page curve analysis \cite{Page:1993df,Almheiri:2019hni,Almheiri:2019psf,Penington:2019npb}.

The first hint of the error correcting structure of the semiclassical limit of AdS/CFT lies within the Ryu--Takayanagi formula \cite{Ryu:2006bv,Ryu:2006ef}. More precisely, each region $A$ of the boundary CFT in a given state that possesses a semiclassical dual corresponds to a bulk region called the \textit{entanglement wedge}, whose geometry and field content are encoded solely in $A$. Given the bulk geometry, the entanglement wedge for a given boundary region is determined through a minimization procedure: it is delimited by a \textit{quantum extremal surface} homologous to the boundary region, and has the property of being an extremum of the generalized entropy of the boundary state. More formally, it is the bulk surface ($a$) homologous to the boundary region that maximizes the quantity
\begin{equation}
S_{gen}:=\frac{1}{4G_N}\text{Area}(a)+S_a(\rho_{bulk}),
\end{equation}
where $\rho_{bulk}$ is the state that describes the system in the semiclassical theory (i.e. the code subspace).

It is important to note that the entanglement wedge can be much larger than the causal wedge, which is the bulk region that is reconstructible through the  Hamilton--Kabat--Lifschytz--Lowe protocol \cite{Hamilton:2005ju,Hamilton:2006az}. For example, this is possible in the case of disconnected subregions of the boundary. In this aspect, entanglement wedge reconstruction is a highly nontrivial and unexpected property of the semiclassical limit of AdS/CFT.

The validity of the Ryu--Takayanagi formula is equivalent to several other claims including complementary recovery \cite{Dong:2016eik}. If the boundary theory is in a pure state and $a$ is the entanglement wedge of a boundary region $A$, then semiclassical operators in $a$ can be reconstructed from operators in $A$ and semiclassical operators in the bulk complement $\overline{a}$ can be reconstructed from the boundary complement $\overline{A}$. This is the key ingredient needed for establishing the connection between entanglement wedge reconstruction and the Ryu--Takayanagi formula in the context of finite-dimensional Hilbert spaces. In this setting, together with the relative entropy equivalence between bulk and boundary derived by Jafferis--Lewkowycz--Maldacena--Suh (JLMS) \cite{Jafferis:2015del}, it has been proven that entanglement wedge reconstruction, the Ryu--Takayanagi formula, and the JLMS formula are all equivalent \cite{Harlow:2016vwg}.

Formulating a rigorous analog of the Ryu--Takayanagi formula in the context of infinite-dimensional Hilbert spaces is hard, as it needs to be regulated, and the absence of tracial states on certain von Neumann algebras makes it particularly tricky to even define von Neumann entropy. However, its equivalent counterparts, namely, complementary recovery and the conservation of relative entropy (and modular flow) between the bulk and the boundary, can still be formulated in the context of infinite-dimensional Hilbert spaces, and are shown to be equivalent under some assumptions in \cite{Kang:2018xqy}. The required assumptions for this exact relation to hold have been extended to more settings in \cite{Gesteau:2020rtg,Faulkner:2020hzi}. It is important to note that this exact relation between entanglement wedge reconstruction and the relative entropy equivalence between bulk and boundary in infinite-dimensional Hilbert spaces is relevant for the case of an actual boundary conformal field theory, utilizing an operator-algebraic perspective.

The analysis in \cite{Kang:2018xqy} relies on von Neumann algebras acting on Hilbert spaces. The code subspace and the physical Hilbert space are embedded into one another via an isometry. The fundamental objects of the theory then become the algebra of CFT observables in a certain boundary region $M_{phys}$ acting on the physical Hilbert space, and an algebra of effective field theory observables $M_{code}$ acting on the code subspace. It has recently been shown that the underlying structure associated to exact entanglement wedge reconstruction was that of a net of conditional expectations between boundary and bulk local algebras \cite{Faulkner:2020hzi}. In this context, the equivalence between the conservation of modular flows and complementary recovery for a state in the code subspace naturally arises as a consequence of Takesaki's theorem.

A striking feature of the entanglement wedge is that it is a state-dependent object. This is due to the presence of the bulk entropy term in the quantum extremal surface formula. In the case of two almost equally contributing local extrema of the generalized entropy \cite{Akers:2019wxj}, or of a large amount of semiclassical entanglement in the bulk \cite{Hayden:2018khn}, this bulk entropy term may become dominant and create large variations of the entanglement wedge, even within the same code subspace. This can happen in situations when the bulk contains a black hole such that there exist enough gravitons giving rise to such entanglement wedge jumps. In such cases, the JLMS formula \cite{Jafferis:2015del} has to be corrected to include the difference in generalized entropies of the two bulk states in their relative entanglement wedges \cite{Dong:2016eik}.

There is an important subtlety in studying some crucial features of the emergence of the bulk in AdS/CFT: the quantum extremal surface formula can only be exact approximately, albeit to nonperturbative order in $G_N$ \cite{Hayden:2018khn}. This subtlety is essential for resolving some apparent paradoxes in AdS/CFT, such as aspects of the information paradox \cite{Penington:2019npb} and non-additivity of entanglement wedges \cite{Kelly:2016edc}. For studying black holes in the $G_N\rightarrow 0$ limit, we allow the code subspace to have an arbitrarily large dimension. In this setting, given a boundary region, a code subspace operator can only be reconstructed in a state-independent manner if it is in its entanglement wedge for both pure and mixed states in the code subspace \cite{Hayden:2018khn}. The intersection of all these entanglement wedges is called the \textit{reconstruction wedge}. In particular, the black hole interior, which lies outside of the reconstruction wedge, can only be reconstructed in a state-dependent manner.

Subtleties about approximate recovery that include non-perturbative gravitational errors have thus far only been considered in finite-dimensional toy models; however, we eventually need to deal with the case of local algebras on an infinite-dimensional boundary Hilbert space. Here, we intend to formulate an approximate version of infinite-dimensional holographic recovery. We first focus on the case of the reconstruction wedge. More specifically, we prove that if a quantum extremal surface formula with an inclusion of non-pertubative gravitational errors can be true in all entanglement wedges, then the bulk algebra of the reconstruction wedge is reconstructable from the boundary up to bounded small nonperturbative gravitational corrections. Putting everything together, we present Theorem \ref{thm:main} that captures the precise result. 
We note that we use \emph{modified-JLMS formula} to depict relative entropy conservation together with generalized entropy terms to deal with jumps in quantum extremal surfaces.

\begin{thm}
Let $\mathcal{H}_{code}$ and $\mathcal{H}_{phys}$ be two Hilbert spaces, $V:\mathcal{H}_{code}\rightarrow \mathcal{H}_{phys}$ be an isometry, and $\mathcal{H}^\ast_{code}$ be any finite-dimensional Hilbert space of dimension smaller or equal to the one of $\mathcal{H}_{code}$. Let $M_{A}$ be a von Neumann algebra on $\mathcal{H}_{phys}$. To each normal state $\omega$ in $\mathcal{B}(\mathcal{H}_{code}\otimes\mathcal{H}^\ast_{code})_\ast$, we associate two entanglement wedge von Neumann algebras $M_{EW(\omega, A)}$ and $M_{EW(\omega, \overline{A}\cup R)}$ of operators on $\mathcal{H}_{code}\otimes\mathcal{H}_{code}^\ast$, such that $M_{EW(\omega, A)}\subset\mathcal{B}(\mathcal{H}_{code})\otimes Id$ and $M_{EW(\omega, \overline{A}\cup R)}\subset M_{EW(\omega, A)}^\prime$. Let
\begin{align*}
    M_{a}:=\underset{\omega}{\bigcap}\;M_{EW(\omega, A)}
\end{align*} 
be the reconstruction wedge von Neumann algebra on $\mathcal{H}_{code}$, and suppose that $M_{a^\prime}$, the relative commutant of $M_a$ in $\mathcal{B}(\mathcal{H}_{code})\otimes Id$, is a product of type $I$ factors. Suppose that for all choices of $\mathcal{H}^\ast_{code}$ and all pairs of states $\rho,\omega$ in $\mathcal{B}(\mathcal{H}_{code}\otimes\mathcal{H}^\ast_{code})_\ast$ such that $S_{\rho,\omega}(M_{EW(\omega, \overline{A}\cup R)})$ is finite, we have the following modified-JLMS condition: 
\begin{align*}
|S_{\rho\circ (\mathcal{E}^c\otimes Id),\omega\circ (\mathcal{E}^c\otimes Id)}(M_A^\prime\otimes\mathcal{B}(\mathcal{H}_{code}^\ast))-S_{\rho,\omega}(M_{EW(\omega, \overline{A}\cup R)})&\\
+S_{gen}(\rho,EW(\rho, \overline{A}\cup R))-S_{gen}(\rho,EW(\omega, \overline{A}\cup R))|&\ \leq\ \varepsilon,
\end{align*}
where $\mathcal{E}$ and $\mathcal{E}^c$ refer to the respective restrictions of $A\mapsto V^\dagger A V$ to $M_A$ and $M^\prime_A$, and the function $S_{gen}(\rho,EW(\omega, \overline{A}\cup R))$ depends only on the restrictions of $\rho$ and $\omega$ to $M_{a^\prime}\otimes\mathcal{B}(\mathcal{H}_{code}^\ast)$.
Then, there exists a quantum channel $\mathcal{R}:M_{a}\rightarrow M_{A}$ such that
\begin{align*}
    \|\mathcal{E}\circ\mathcal{R}-Id_{M_{a}}\|_{cb} \leq 2(2\varepsilon)^\frac{1}{4}.
\end{align*}
For finite-dimensional $\mathcal{H}_{code}$, we find a bound of $2\sqrt{2\sqrt{2\varepsilon}}$ instead.
\label{thm:main}
\end{thm}
We then move on to the state-dependent case, and explain how to extend the state-dependent $\alpha$-bit reconstruction of the black hole interior \cite{Hayden:2018khn} to the infinite-dimensional case. We also compare our approach through complementary recovery to the recent extension \cite{Faulkner:2020iou} of the twirled Petz map \cite{Cotler:2017erl} to the infinite-dimensional case. While the twirled Petz map provides an explicit recovery channel, the approximate recovery it gives is only controlled by a bound involving the operator norm, rather than the completely bounded norm that we get from our approach. Moreover, it is unclear whether the twirled Petz map still provides a valid reconstruction at subleading order in $G_N$, whereas the error in our approach is nonperturbative.

Taking this operator algebraic framework \cite{Kang:2018xqy,Kang:2019dfi,Faulkner:2020hzi,Faulkner:2020iou,Faulkner:2020kit,Harlow:2016vwg,Gesteau:2020rtg,Gesteau:2020wrk,Gesteau:2020hoz,Witten:2018zxz} on understanding bulk reconstruction, modular structure, and relative entropy equivalence between bulk and boundary, our results expand the dictionary between holography and operator algebras. We summarize this in Table \ref{tb:dictionary}, which also effectively captures the emerging quantum error correcting structure.

\begin{table}[H]
\centering
\begin{threeparttable}
\begin{tabular}{c c}
\toprule
    \textbf{Holography} & \textbf{Operator Algebras}\\\midrule
    Physical operators & Von Neumann algebras $M_{phys}$\\
    Logical operators & Von Neumann algebra $M_{code}$\\
    Projection onto the logical operators & Conditional expectation\\
    Exact bulk reconstruction & Invariance under conditional expectation\\
    JLMS formula & Takesaki theorem\\
    Petz map & Generalized conditional expectation\\\
    State-independent bound & Completely bounded norm bound\\
    (Approximate) complementary recovery & ($\varepsilon$-) correctability/privacy duality\\
    Twirled Petz map & Faulkner--Hollands map
\\\bottomrule
\end{tabular}
\end{threeparttable}
\caption{This provides a dictionary between concepts in holography and their operator algebraic counterparts.}
\label{tb:dictionary}
\end{table}

At the end of the paper, we expand on how our algebraic approach enables us to capture nontrivial features of the bulk, and propose some future research directions. In particular, we comment on the possibility of bridging the quantum error correction perspective to the semiclassical limit of the bulk, with the recent advances involving averaged theories and gravitational path integrals. We further analyze the modified-JLMS condition required for our theorem in more detail and show how our result can be relevant to the case of an evaporating black hole.

The remaining parts of the paper are organized as follows. In Section \ref{sec:preliminaries}, we first provide a toolkit of relevant mathematical concepts for the purpose of this paper, including conditional expectations, modular flow, relative entropy from Araki's perspective, the Petz map, and basic Tomita--Takesaki theory. In Section 3, we provide background on existing results for exact entanglement wedge reconstruction in infinite dimensions. In Section 4, we introduce the notions of approximate privacy and correctability for quantum channels, as well as a fundamental result that relates them. In Section 5, we utilize this result to prove that an approximate modified-JLMS formula implies approximate complementary recovery in the reconstruction wedge. In Section 6, we move on to the state-dependent case, beyond the reconstruction wedge. We describe the properties of black hole $\alpha$-bits in the case of an infinite-dimensional boundary Hilbert space, and we study the fundamental differences between our framework and the use of a universal recovery channel. In particular, we find that our framework, although non-constructive, is more robust, and valid up to nonperturbative errors in $G_N$. In Section 7, we explain how this perspective can help understanding various physical settings from our result, and provide some promising research directions. In particular, we explain how to connect our algebraic approach with the gravitational path integral and probing averages over theories. We study in detail the modified-JLMS condition, which we utilized as an assumption to our main theorem, and show the relevance of our framework in the case of an evaporating black hole.

\section{Preliminaries: algebraic perspective}
\label{sec:preliminaries}

In this section, we describe various mathematical objects that are used in this paper for studying entanglement wedge reconstruction. We also give some important properties of these objects. 

\subsection{Conditional expectations}

For holographic theories, the algebra of observables of quantum field theory on a fixed spacetime geometry is usually described as a von Neumann algebra acting on a Hilbert space, which is then isometrically embedded in a bigger Hilbert space onto which a bigger von Neumann algebra acts.

Motivated by this picture, we want to have an appropriate way to project any element of a von Neumann algebra onto a smaller one. The right object to consider for such an operation is a \textit{conditional expectation}, which can be defined as the following. 

\begin{defn}
Let $M$ be a von Neumann algebra and $N$ be a von Neumann subalgebra of $M$. A conditional expectation from $M$ to $N$ is a positive linear map
\begin{align*}
    E:M\to N
\end{align*}
satisfying, for $A\in M$ and $B,C\in N$,
\begin{align*}
    E(Id)=Id,\quad E(BAC)=BE(A)C .
\end{align*}
\end{defn}

One should think of a conditional expectation as the abstract equivalent of ``tracing out degrees of freedom." It takes into account the fact that an observer may only have an access to part of the information contained in a quantum system - in our case, the code subalgebra. 

The properties of conditional expectations, and how they interact with the states of the theory, will be at the heart of our discussion on entanglement wedge reconstruction. Indeed, as we will shortly see, in the exact picture, the holographic map can be identified with a conditional expectation \cite{Faulkner:2020hzi}.
\subsection{Tomita--Takesaki theory and modular flow}

Tomita--Takesaki theory is fundamental in the theory of von Neumann algebras. We briefly define the main objects appearing in this theory, and state the Tomita--Takesaki theorem. We utilize these concepts for our purposes to study relative entropies and the modular flow.

\begin{defn}
	\label{def:cyc}
	A vector $\ket{\Psi} \in \calh$ is {\em cyclic} with respect to a von Neumann algebra $M$ acting on a Hilbert space $\mathcal{H}$ when the vectors $\calo\ket{\Psi}$ for $\calo \in M$ form a dense set in $\calh$. 
\end{defn}
\begin{defn}
	\label{def:sep}
	A vector $\ket{\Psi} \in \calh$ is {\em separating} with respect to a von Neumann algebra $M$ when $\calo\ket{\Psi} = 0 \implies \calo = 0$ for $\calo \in M$. 
\end{defn}
It is important to note that a cyclic and separating state can be interpreted as a state that sufficiently entangles the observables of the algebra $M$ with the rest of the system. In finite-dimensions, a state being cyclic and separating for a given algebra boils down to its restriction to the algebra being a thermal density matrix.

We remind the reader that states on von Neumann algebras can be defined as norm-continuous positive linear functionals of norm 1 directly.
\begin{defn}
A state $\omega$ on a von Neumann algebra $M$ is \textit{normal} if for any uniformly bounded, monotone, increasing net of positive operators $H_\alpha\in M$, $$\omega(\underset{\alpha}{\mathrm{sup}}\; H_\alpha)=\underset{\alpha}{\mathrm{sup}}\;\omega(H_\alpha).$$
\end{defn}
The predual $M_\ast$ of the von Neumann algebra $M$ is then spanned by the normal states.
\begin{defn}
A state $\omega$ on a von Neumann algebra $M$ is \textit{faithful} if for $A\in M$, $$\omega(A^\ast A)=0\implies A=0.$$
\end{defn}
It is easy to see that a separating vector for the von Neumann algebra $M$ induces a faithful state on $M$.

With these refined properties of states formulated in an algebraic manner, we now further define (relative) modular operators, first on a certain subset of the Hilbert space.
\begin{defn}
Let $\ket{\Psi},\ket{\Phi} \in \calh$ and $M$ be a von Neumann algebra. 
The {\em relative Tomita operator} is the operator $S_{\Psi | \Phi}$ defined by
\begin{equation*}
S_{\Psi|\Phi} \ket{x} := \ket{y}
\end{equation*}
for all sequences $\{\calo_n\} \in M$ for which the limits $\ket{x} = \lim_{n \rightarrow \infty} \calo_n \ket{\Psi}$ and $\ket{y} = \lim_{n \rightarrow \infty} \calo_n^\dagger \ket{\Phi}$ both exist.
\end{defn}

If $\ket{\Psi}$ is cyclic and separating with respect to the von Neumann algebra $M$, the relative Tomita operator can be defined on a dense subset of the Hilbert space.

\begin{thm}[\cite{Jones-vNalg}, p.94]
	Let $\ket{\Psi},\ket{\Phi} \in \calh$ be cyclic and separating vectors with respect to a von Neumann algebra $M$. Let $S_{\Psi|\Phi}$ and $S^\prime_{\Psi|\Phi}$ be the relative Tomita operators respectively defined with respect to $M$ and its commutant $M^\prime$. Then
	\begin{equation}
	S_{\Psi|\Phi}^\dagger = S^\prime_{\Psi|\Phi}, \ S_{\Psi|\Phi}^{\prime \, \dagger} = S_{\Psi|\Phi}.
	\end{equation}
\end{thm}

\begin{defn}
Let $S_{\Psi|\Phi}$ be the relative Tomita operator for two vectors $\ket{\Psi},\ket{\Phi} \in \calh$, which are cyclic and separating with respect to a von Neumann algebra $M$. The {\em relative modular operator} is $$\Delta_{\Psi|\Phi} := S_{\Psi|\Phi}^\dagger S_{\Psi|\Phi}.$$
\end{defn}

When $\ket{\Psi}$ and $\ket{\Phi}$ are the same vector, the relative Tomita operator specializes into a Tomita operator as following.

\begin{defn}
	Let $M$ be a von Neumann algebra on $\calh$ and $\ket{\Psi}$ be a cyclic and separating vector for $M$. The \emph{Tomita operator} $S_\Psi$ is
	$$ S_\Psi := S_{\Psi | \Psi},$$ where $S_{\Psi | \Psi}$ is the relative modular operator defined with respect to $M$. The \emph{modular operator} $\Delta_\Psi = S_\Psi^\dagger S_\Psi$ and the \emph{modular conjugation} $J_\Psi$ are the operators that appear in the polar decomposition of $S_\Psi$ such that
	$$S_\Psi = J_\Psi \Delta_{\Psi}^{1/2}.$$
\end{defn}

We note that the modular operator given by a state defined by a thermal density matrix in finite-dimensions has a simple expression: it is just the thermal density matrix itself. In particular, its logarithm coincides with the familiar notion of modular Hamiltonian.

The Tomita--Takesaki theorem is a powerful tool to study properties of von Neumann algebras, which is given as the following.

\begin{thm}[Tomita--Takesaki \cite{Tomita-Takesaki}]
	\label{thm:modflow}
	Let $M$ be a von Neumann algebra on $\calh$ and let $\ket{\Psi}$ be a cyclic and separating vector for $M$. Then \begin{itemize}
		\item $J_\Psi M J_\Psi = M^\prime.$
		\item $\Delta_\Psi^{it} M \Delta_\Psi^{-it} = M \quad \forall t \in \mathbb{R}$.
	\end{itemize} 
\end{thm}

This theorem shows that, associated to each cyclic separating state, there is a canonical ``time evolution" given by the modular flow. This flow can be defined through conjugating an algebra element by imaginary powers of the modular operator, and its physical interpretation is that the restriction of the state $\ket{\psi}$ to $M$ is at thermal equilibrium with respect to its modular flow \cite{Connes:1994hv}. In quantum field theory, the Bisognano--Wichmann theorem \cite{Haag} shows that the modular flow of the vacuum generates null boosts for Rindler wedges. The fact that the vacuum is at thermal equilibrium with respect to these null boosts is then related to the Unruh effect.

We emphasize that we do not require any additional structure, but rather, let the operator-algebraic analysis to provide all the necessary tools. However, our setup can be applied to more structured particular cases as well where there exists a more general definition of modular objects for non-cyclic and non-separating vectors; see for example \cite{Ceyhan:2018zfg}.

\subsection{The Connes cocycle and Araki's relative entropy}

Based on the notion of modular flow, we can now define the Connes cocycle of two states $\psi$ and $\varphi$, with $\varphi$ faithful and normal.

\begin{defn}
Let $\psi$ and $\varphi$ be two states on a von Neumann algebra $M$, with $\varphi$ faithful and normal. The Connes cocycle of $\psi$ and $\varphi$ is defined as $$
[D\psi:D\varphi]_t:=\Delta_{\psi|\varphi}^{it}\Delta_{\varphi}^{-it}.
$$
\end{defn}

It is a highly nontrivial result that the Connes cocycle of two states on $M$ is always in $M$. In other words, the difference between modular flows of two states is always an inner automorphism. The Connes cocycle has a nice geometric interpretation in AdS/CFT \cite{Bousso:2020yxi}, and it would be nice to understand it in more detail in the context of our framework.

Another useful quantity to define from Tomita--Takesaki theory is Araki's relative entropy, which is provided in Definition \ref{def:relent}.

\begin{defn}[Araki \cite{Araki}] \label{def:relent}
Let $\ket{\Psi},\ket{\Phi} \in \calh$ and $\ket{\Psi}$ be cyclic and separating with respect to a von Neumann algebra $M$. Let $\Delta_{\Psi|\Phi}$ be the relative modular operator. The {\em relative entropy} with respect to $M$ of $\ket{\Psi}$ is
	\begin{equation*} 
	\cals_{\Psi|\Phi}(M) = - \braket{\Psi|\log \Delta_{\Psi | \Phi}|\Psi}. 
	\end{equation*} 
\end{defn}

It is important to note that Araki's relative entropy can be obtained from the Connes cocycle by the following formula:
\begin{align}
    S(\omega,\varphi)=i\underset{t\rightarrow 0^+}{\mathrm{lim}}\,t^{-1}(\omega([D\varphi:D\omega]_t)-1).
\end{align}

This formula shows that Araki's relative entropy can be thought of as a sort of derivative of an expectation value of the Connes cocycle flow.

Araki's notion of relative entropy will be the main tool coming from Tomita--Takesaki theory that we will require in this paper. Indeed, the equivalence of relative entropies between the bulk and the boundary is one of the equivalent characterizations of exact entanglement wedge reconstruction, as well as the starting point of our theorem, which tackles approximate recovery settings that include nonperturbative gravity corrections. 

It is important to remind the readers that there exists a generalization of the notion of relative entropy that allows to define it even for non-cyclic and non-separating states, and satisfies the Pinsker inequality. The details can be found in \cite{OhyaPetz}.

\subsection{The Petz map from Tomita--Takesaki theory}

A map that is often proposed for entanglement wedge reconstruction is the Petz map \cite{Cotler:2017erl,Chen:2019gbt}. This map has a particularly nice property as a natural map between a von Neumann algebra and a von Neumann subalgebra leaving a given state $\varphi$ invariant. 

\begin{defn}
Let $M_{code}\subset M_{phys}$ be two von Neumann algebras acting on a Hilbert space $\mathcal{H}$. Let $\ket{\Phi}$ be a cyclic separating vector for $M_{phys}$, and $J$ the associated modular conjugation. Similarly, let $J_0$ be the modular conjugation for $M_{code}$ acting on $\overline{M_{code}\ket{\Phi}}$. Let $P$ be the orthogonal projection onto $\overline{M_{code}\ket{\Phi}}$. Then, the Petz map $E_\varphi$ associated to $\ket{\Phi}$ is given by 
$$E_\varphi(A):=J_0PJAJPJ_0.$$
\end{defn}

It is important to note that this $E_\varphi$ is a map from $M_{phys}$ to $M_{code}$ and satisfies 
\begin{align}
    \varphi\circ E_\varphi=\varphi.
\end{align}
However, nothing tells us that the Petz map will follow the axioms of a conditional expectation!\footnote{One should approach with caution the misleading terminology ``$\varphi$-conditional expectation" for the Petz map, which is often used in the mathematical literature.} In fact, we explain a necessary and sufficient condition for it to be a conditional expectation in Section \ref{sec:condexpectation} via a theorem of Takesaki. We want to emphasize this aspect as the cases where the Petz map is a conditional expectation will coincide with the cases of exact entanglement wedge reconstruction in our discussion. When taking into account corrections from gravitational effects that are nonperturbative in $G_N$, the ``exact" entanglement wedge reconstruction breaks down. In this scenario with gravitational corrections, the situation will be much more subtle and the Petz map will no longer be a good choice for entanglement wedge reconstruction in general, although we will briefly comment on the efficacy of its twirled version in Section \ref{sec:recoverychannel}.

\section{Exact infinite-dimensional entanglement wedge reconstruction}

We review the code subspace formalism for the exact correspondence between entanglement wedge reconstruction and relative entropy equivalence. This perspective is first introduced in \cite{Harlow:2016vwg}, later extended to infinite-dimensional settings in \cite{Kang:2018xqy,Kang:2019dfi,Faulkner:2020hzi}, and this technology has even been further expanded to state-dependent contexts in \cite{Gesteau:2020rtg}. Compositing these together and bringing it to a setting of our interest, we provide a summarizing theorem that prescribes the mathematical structure underlying this formalism, including conditional expectations, the Petz map, and Takesaki's theorem.

\subsection{Entanglement wedge reconstruction in the code subspace formalism}

The Ryu--Takayanagi formula, which relates the generalized entropy of a semiclassical bulk state of AdS/CFT to the entanglement entropy of the corresponding boundary state, is equivalent to relative entropy conservation between bulk and boundary \cite{Jafferis:2015del} and the quantum error correcting structure of the semiclassical limit of AdS/CFT, which allows us to see the space of states that represent effective field theory on a fixed curved background in the bulk as a code subspace of the space of boundary states.

In the modern information-theoretic formalism, this code subspace of the theory is seen as a vector space isometrically embedded into the physical Hilbert space of quantum gravity. In this setup, bulk reconstruction is seen as the existence of operators acting on the physical Hilbert space that reproduce the action of operators acting on the code subspace. The equivalence between the complementary recovery and the conservation of relative entropy (in the context of infinite-dimensions) has been proven in increasingly general cases \cite{Kang:2018xqy,Kang:2019dfi,Gesteau:2020rtg,Faulkner:2020hzi}. Here, we give a summary theorem that comprises all these results in the context of von Neumann algebras.

\begin{thm}
\label{thm:summary}
Let $u:\mathcal{H}_{code}\rightarrow\mathcal{H}_{phys}$ be an isometric map of Hilbert spaces, let $M_{phys}$ be a von Neumann algebra on $\mathcal{H}_{phys}$ and $M_{code}$ be a von Neumann algebra on $\mathcal{H}_{code}$. Suppose that there exists $\ket{\psi}\in\mathcal{H}_{code}$ such that $u\ket{\psi}$ is cyclic and separating with respect to $M_{phys}$. Consider two following statements:

    \noindent (i) There exist unital normal injective $^\ast$-homomorphisms $\iota:M_{code}\rightarrow M_{phys}$ and $\iota^\prime:M_{code}^\prime\rightarrow M_{phys}^\prime$ such that for $O,O^\prime\in M_{code},M^\prime_{code}$, $$u O=\iota(O)u,\;uO^\prime=\iota^\prime(O^\prime)u.$$
    (ii) If $\ket{\varphi}$ and $\ket{\psi}$ are two vectors in $\mathcal{H}_{code}$, then 
    \begin{align}
    \label{entropy}
    S_{\psi,\varphi}(M_{code})=S_{u\psi,u\varphi}(M_{phys}),\quad
    S_{\psi,\varphi}(M^\prime_{code})=S_{u\psi,u\varphi}(M^\prime_{phys}).
    \end{align}

Then, $(i)\Rightarrow (ii)$.

If we further assume that the set of cyclic and separating vectors with respect to $M_{code}$ is dense in $\mathcal{H}_{code}$ and $u$ maps cyclic and separating states for $M_{code}$ to cyclic and separating states for $M_{phys}$, then $(i)\Leftrightarrow (ii)$.
\end{thm}

Note that in \cite{Kang:2018xqy}, it is not explicitly shown that $\iota$ and $\iota^\prime$ define unital, normal, injective $^\ast$-homomorphisms. However, it is a straightforward consequence of the fact that $u\ket{\psi}$ is cyclic and separating with respect to $M_{phys}$ and $M^\prime_{phys}$: the separating property directly implies homomorphism properties as well as unitality and injectivity, while normality follows from Proposition 2.5.11 of \cite{Anantharaman}.

In the next subsection, we show how Theorem \ref{thm:summary} is deeply related to some well-explored concepts in the theory of von Neumann algebras, such as conditional expectations and an important theorem of Takesaki \cite{Takesaki:72}.

\subsection{Link to conditional expectations and Takesaki's theorem}
\label{sec:condexpectation}

Under the assumptions Theorem \ref{thm:summary}, we can infer that the reconstruction operators $\iota(O)$s form a von Neumann subalgebra of $M_{phys}$.
We suppose that the assumptions of Theorem \ref{thm:summary} are satisfied. Utilizing this setup, we now turn to constructing a conditional expectation from $M_{phys}$ onto $\iota(M_{code})$. 

In the proofs of the various results available on exact entanglement wedge reconstruction in infinite-dimensions \cite{Kang:2018xqy,Faulkner:2020hzi}, a crucial lemma showed that if $O\in M_{phys}$, then $u^\dagger Ou\in M_{code}$. It follows that the map 
\begin{align}
    E:O\longmapsto\iota(u^\dagger Ou)
\end{align}
is well-defined from $M_{phys}$ to $\iota(M_{code})$ and satisfies the conditional expectation property. Indeed, we see that the identity is preserved under this map,
\begin{align}
    E(Id)=\iota(u^\dagger u)=Id,
\end{align}
and if $A\in M_{phys}$ and $B,C\in M_{code}$,
\begin{align}
    E(\iota(B)A\iota(C))=\iota(u^\dagger\iota(B)A\iota(C)u)=\iota(Bu^\dagger AuC)=\iota(B)\iota(u^\dagger Au)\iota(C)=\iota(B)E(A)\iota(C).
\end{align}
Moreover, $E$ leaves any state of the form $u\ket{\varphi}$ invariant, if $\ket{\varphi}$ is cyclic and separating for $\mathcal{M}_{code}$. For such cases, we see that $u\ket{\varphi}$ is cyclic and separating for $\mathcal{M}_{phys}$ and, for $A\in M_{phys}$, 
\begin{align}
    \bra{\varphi}u^\dagger E(A)u\ket{\varphi}=\bra{\varphi}u^\dagger u u^\dagger A u\ket{\varphi}=\bra{\varphi}u^\dagger Au\ket{\varphi}.
\end{align}

This remark directs us towards a more conceptual understanding of entanglement wedge reconstruction. Indeed, the isometry $u$ provides us with a conditional expectation, that projects $M_{phys}$ onto $\iota(M_{code})$ and preserves the states in the code subspace. This is the definition of a \textit{sufficient} von Neumann algebra inclusion, which is the cornerstone of the framework of exact entanglement wedge reconstruction.

Interestingly, this structure can be interpreted as a theorem of Takesaki in \cite{Takesaki:72}. In fact, this theorem shows the equivalence between the existence of a state-preserving conditional expectation, from a von Neumann algebra onto one of its subalgebras, and the stabilization of this subalgebra under the modular flow of the state. We rephrase this theorem as Theorem \ref{thm:Takesaki} to be best applied to our physical setup.

\begin{thm}[Takesaki \cite{Takesaki:72}]
\label{thm:Takesaki}
Let $\varphi$ be a faithful normal state on a von Neumann algebra $M$, and let $N$ be a von Neumann subalgebra of $M$. Then, there exists a faithful normal conditional expectation $E:M\rightarrow N$ such that $\varphi\circ E=\varphi$ if and only if $\sigma_t(N)\subset N$, where $\sigma_t$ is the modular flow associated to $\varphi$ for the von Neumann algebra $M$. Moreover, the conditional expectation $E$ which realizes the sufficiency condition is unique, and coincides with the Petz map of $\varphi$. 
\end{thm}

From this perspective, in the code subspace formalism, all states of the form $u\ket{\psi}$, with $\psi\in\mathcal{H}_{code}$, are preserved under the conditional expectation $E$. Then, in particular, $E$ corresponds to the Petz map of these states, and the algebra $\iota(M_{code})$ is preserved under modular flow. Under the conditions of Takesaki's theorem, it is then possible to prove the following: for all faithful normal states $\omega,\varphi$ on $M_{phys}$ induced by vectors of the form $u\ket{\omega},u\ket{\varphi}\in\mathcal{H}_{phys}$ with $\ket{\omega},\ket{\varphi}\in\mathcal{H}_{code}$ \cite{OhyaPetz},

(i) Connes cocycles are conserved:
\begin{align}
    [D\omega:D\varphi]_t=[D\omega|_{\iota\left(M_{code}\right)}:D\varphi|_{\iota\left(M_{code}\right)}]_t.
\end{align}

(ii) Relative entropies are conserved: 
\begin{align}
    \mathcal{S}_{\omega,\phi}(M_{phys})=\mathcal{S}_{\omega|_{\iota\left(M_{code}\right)},\phi|_{\iota\left(M_{code}\right)}}(\iota(M_{code})).
\end{align}

(iii) Petz maps coincide: 
\begin{align}
    E_\varphi=E_\omega.
\end{align}

\subsection{Limitations of the exact approach}

We briefly comment on how the exact approach is insufficient to capture some crucial aspects of entanglement wedge reconstruction in AdS/CFT. First, one has to bear in mind that the AdS/CFT duality is a large $N$ statement from a CFT perspective: as $N$ is taken to be large but finite, there will be stringy corrections to the ideal picture of a reconstruction through a conditional expectation, and this will make the picture to break down. In other words, making Newton's constant $G_N$ finite breaks the exact quantum error correcting code structure. This is not a mere technicality and it has been pointed out in \cite{Kelly:2016edc,Hayden:2018khn,Penington:2019npb} that handling approximation is crucial for understanding the semiclassical limit of AdS/CFT and accounts for important physical aspects of the correspondence. In particular,
\begin{itemize}
    \item The redundancy in the bulk-to-boundary encoding can only be consistent with the Reeh-Schlieder theorem on the boundary if entanglement wedge reconstruction is approximate (with an error that may be nonperturbative in $G_N$) \cite{Kelly:2016edc}.
    \item For a state-independent recovery to be possible in the approximate setting, a local bulk operator has to be in the entanglement wedge of \textit{all pure and mixed states} of the corresponding boundary region (this is called the \textit{reconstruction wedge}) \cite{Hayden:2018khn}. This feature cannot be seen through an exact approach.
    \item The presence of nonperturbative errors in bulk reconstruction is crucial to the resolution of the black hole information paradox \cite{Penington:2019npb}.
\end{itemize}

Our goal is then to generalize, to an approximate setting, the relation between entanglement wedge reconstruction and the relative entropy equivalence between bulk and boundary. In turn, this requires borrowing and extending the concept of approximate complementary recovery, introduced in the context of AdS/CFT in \cite{Hayden:2018khn}. In the remainder of the paper, we combine this toolkit and the operator-algebraic perspective to extend the results from \cite{Kang:2018xqy,Kang:2019dfi,Gesteau:2020rtg} to an approximate setting.

In order to reach this objective, we first need to introduce some vocabulary regarding correctability and privacy for infinite-dimensional von Neumann algebras, as well as an important result due to Crann--Kribs--Levene--Todorov \cite{CKLT}, which will be at the heart of our proof, and can be seen as an infinite-dimensional generalization of the information-disturbance tradeoff theorem \cite{Kretschmann}. We will present and utilize this machinery to derive an approximate relation between entanglement wedge reconstruction and relative entropy conservation between bulk and boundary.

\section{Private and correctable algebras and complementary recovery}
\label{sec:privacycorrectability}

The strongest argument for entanglement wedge reconstruction is the argument of Dong--Harlow--Wall in \cite{Dong:2016eik}, which derives entanglement wedge reconstruction from the modified-JLMS formula. As demonstrated in \cite{Hayden:2018khn}, there are many subtleties related to the problem of generalizing the argument of \cite{Dong:2016eik} to the approximate case; in particular, an operator can be reconstructed on the boundary in a state-independent way only if it is in the entanglement wedge of all pure and mixed states in the code subspace. If indeed this is the case, then the main ingredient for proving that the modified-JLMS formula implies approximate reconstruction is the information-disturbance tradeoff \cite{Kretschmann}. 

Here, we introduce a framework which allow us to generalize this notion to the case of an infinite-dimensional boundary Hilbert space for the CFT, as expected in any quantum field theoretic setting. The main ingredients for constructing this Hilbert space setting involve the notions of completely bounded norm, (approximately) private and correctable algebras, as well as a theorem by Crann, Kribs, Levene and Todorov \cite{CKLT}. In particular, we take the norm to be the completely bounded norm, as it is the one appropriate for state-independent recovery. In fact, the completely bounded norm takes into account an extra tensor factor representing an auxiliary system, and that is crucial for recreating all pure and mixed states on the code subspace. As such, a completely bounded norm bound will quantify the quality of a recovery quantum channel with respect to all states at the same time. We expect the completely bounded norm to be bounded whenever the we can perform a recovery in a state-independent fashion. We utilize a region called \emph{reconstruction wedge} to describe the appropriate geometric region of the semi-classical gravitational dual where state-independent recovery is possible.

\subsection{The completely bounded norm, privacy, and correctability}

In order to utilize a von Neumann algebra construction for our setup, we need to introduce a new norm for bounded linear maps, which is called \emph{completely bounded norm}.

\begin{defn}
Let $\varphi:M\rightarrow N$ be a bounded linear map between two von Neumann algebras. We define the \textit{completely bounded norm} of $\varphi$ as $$\|\varphi\|_{cb}:=\mathrm{sup}_{k\in\mathbb{N}}\|\varphi\otimes Id_{\mathcal{M}_k(\mathbb{C})}\|.$$ 
If this quantity is finite, then we say that $\varphi$ is \textit{completely bounded}.
\end{defn}
In particular, one can prove that a quantum channel (i.e., a normal, unital, completely positive map) is always completely bounded. Another nice property (called Smith's lemma) is that if $N$ is of the form $\mathcal{M}_n(\mathbb{C})$, then the supremum in the definition can be only taken for $k\leq n$.

The completely bounded norm is the ``Heisenberg picture" equivalent of the notion of diamond norm, which appears more often in the quantum information literature. However, the completely bounded norm is more natural in our operator-algebraic setting, and also reacts better with the conditional expectation-like structure we will use.

We now give the definition of a private algebra with respect to a quantum channel (i.e., a normal, unital, completely positive map).

\begin{defn}
Let $\mathcal{H}$ be a Hilbert space, $M$ be a von Neumann algebra, and $\mathcal{E}:M\rightarrow \mathcal{B}(\mathcal{H})$ a quantum channel. A von Neumann algebra $N\subset\mathcal{B}(\mathcal{H})$ is said to be \textit{private} for $\mathcal{E}$ if $\mathcal{E}(M)\subset N^\prime$. For $\varepsilon>0$, $N$ is said to be $\varepsilon$\textit{-private} for $\mathcal{E}$ if there exists another quantum channel $\mathcal{P}:M\rightarrow \mathcal{B}(\mathcal{H})$ such that $N$ is private for $\mathcal{P}$ and $\|\mathcal{E}-\mathcal{P}\|_{cb}\leq \varepsilon$.
\end{defn}

The interpretation of a private algebra for a given quantum channel is that none of the information it contains is accessible from the domain of the quantum channel. For instance, in an exact idealization of AdS/CFT, if $A$ is a boundary region with entanglement wedge $a$, then the bulk algebra of $a$ is private with respect to the complement of $A$ for the usual boundary to bulk map. We will show that a similar statement holds in the approximate case, in a somewhat smaller region called the \textit{reconstruction wedge}. 

The ``dual" notion of privacy is that of correctability, which is given by the following definition: 

\begin{defn}
Let $\mathcal{H}$ be a Hilbert space, $M$ be a von Neumann algebra, and $\mathcal{E}:M\rightarrow \mathcal{B}(\mathcal{H})$ a quantum channel. A von Neumann algebra $N\subset\mathcal{B}(\mathcal{H})$ is said to be \textit{correctable} for $\mathcal{E}$ if there exists a quantum channel $\mathcal{R}:N\rightarrow M$ such that $\mathcal{E}\circ\mathcal{R}=Id_N.$ For $\varepsilon>0$, $N$ is said to be $\varepsilon$\textit{-correctable} for $\mathcal{E}$ if there exists a quantum channel $\mathcal{R}:N\rightarrow M$ such that $\|\mathcal{E}\circ\mathcal{R}-Id_N\|_{cb}\leq\varepsilon.$
\end{defn}

In the exact setting, correctability of the entanglement wedge $a$ of region $A$ corresponds to the existence of the $\ast$-homomorphism $\iota$ between the algebra $M_{code}$ and $M_{phys}$, respectively corresponding to $a$ and $A$.

Note that both in the definition of a private algebra and in the definition of a correctable algebra, we have resorted to the completely bounded norm, rather than the usual operator norm. Such a distinction is crucial in order for the next theorem to work, and can be traced back to the fact that state-independent reconstruction of a bulk operator is only possible if it is in the entanglement wedge of \textit{both pure and mixed} states in the code subspace.

\subsection{A duality between privacy and correctability}

If $\mathcal{H}$ is a Hilbert space, $M$ a von Neumann algebra, and $\mathcal{E}:M\rightarrow \mathcal{B}(\mathcal{H})$ a quantum channel, recall that by the Stinespring dilation theorem \cite{Stinespring}, there exist a Hilbert space $\mathcal{K}$, a representation $\pi$ of $M$ on $\mathcal{K}$ and an isometry $V:\mathcal{H}\rightarrow\mathcal{K}$ such that for $A\in M$,
\begin{align}
    \mathcal{E}(A)=V^\dagger\pi(A)V.
\end{align} 
Following \cite{CKLT}, we call the datum $(\pi,V,\mathcal{K})$ a Stinespring triple for $\mathcal{E}$. Given a Stinespring triple, it is natural to define the complementary channel of a channel $\mathcal{E}$ on the commutant of $\pi(M)$, which we give below as Definition \ref{def:complementarychannel}.

\begin{defn}
\label{def:complementarychannel}
Let $\mathcal{H}$ be a Hilbert space, $M$ be a von Neumann algebra, and $\mathcal{E}:M\rightarrow \mathcal{B}(\mathcal{H})$ a quantum channel. Let $(\pi,V,\mathcal{K})$ be a Stinespring triple for $\mathcal{E}$. We define the complementary channel for the triple $(\pi,V,\mathcal{K})$,  $\mathcal{E}^c:\pi(M)^\prime\rightarrow\mathcal{B}(\mathcal{H})$, by 
\begin{align*}
    \mathcal{E}^c(A):=V^\dagger AV.
\end{align*} 
\end{defn}
With this definition in hand, we now have the tools to describe the main theorem of this section, which is the crucial ingredient of our discussion of approximate recovery in the reconstruction wedge, by Crann, Kribs, Levene and Todorov: 

\begin{thm}[Crann--Kribs--Levene--Todorov\cite{CKLT}]
\label{thm:CKLT}
Let $\mathcal{H}$ be a Hilbert space, $M$ be a von Neumann algebra on a Hilbert space $\mathcal{K}$, and $\mathcal{E}:M\rightarrow \mathcal{B}(\mathcal{H})$ a quantum channel. If a von Neumann subalgebra $N\subset \mathcal{B}(\mathcal{H})$ is $\varepsilon$-private (resp. $\varepsilon$-correctable) for $\mathcal{E}$, then it is $2\sqrt{\varepsilon}$-correctable (resp. $8\sqrt{\varepsilon}$-private) for any complementary channel to $\mathcal{E}$.
\end{thm}

This result is far from trivial, and states that there is an equivalence between being approximately private (i.e., almost no information in the bulk region is accessible from the boundary region), and having an approximately correctable commutant (i.e., almost all information in the bulk region can be reconstructed in the complementary boundary region), in a Stinespring representation. 

We now pause for a moment to consider the specific case of $\varepsilon=0$, before moving on to the approximate case. The case $\varepsilon=0$ shows that being \textit{exactly} private for a given quantum channel implies being \textit{exactly} correctable for any complementary channel. In holographic settings, we retrieve the result of Dong--Harlow--Wall for the case of finite-dimensional Hilbert spaces \cite{Dong:2016eik}. This provides a check on our infinite-dimensional construction as a natural generalization of the gravitational recovery. 

Indeed, in the code subspace formalism, consider a von Neumann algebra $M_A$ acting on $\mathcal{H}_{phys}$, and its commutant $M^\prime_A$. Introduce the code subspace injection $V:\mathcal{H}_{code}\rightarrow\mathcal{H}_{phys}$. Then we denote $M_a$ and $M_a^\prime$ as the bulk entanglement wedge algebras that correspond to $A$ and its complement. Since relative entropy has to be conserved, for any two states $\rho$ and $\varphi$ with finite relative entropy in the code subspace, 
\begin{align}
    S_{\rho\circ \mathcal{E}^c,\;\varphi\circ\mathcal{E}^c}(M^\prime_{A})=S_{\rho,\varphi}(M^\prime_{a}).
\end{align}
Now suppose that there exists a normal norm-one projection $\mathcal{P}_{a^\prime}$ from $\mathcal{B}(\mathcal{H}_{code})$ onto $M^\prime_a$. In particular, for $\varphi=\rho\circ\mathcal{P}_{a^\prime}$, the relative entropy conservation formula gives 
\begin{align}
    S_{\rho\circ \mathcal{E}^c,\;\rho\circ\mathcal{P}_{a^\prime}\circ\mathcal{E}^c}(M^\prime_{A})=0.
\end{align}
This implies that the algebra $M_a$ is private for the quantum channel $\mathcal{E}^c$, so it is correctable for the quantum channel $\mathcal{E}$. The same argument holds for $M_a^\prime$ as long as there exists a normal norm one projection $\mathcal{P}_a$.\footnote{We will see that this is equivalent to $M_a^\prime$ being a product of type $I$ factors.} Thus, we have proven complementary recovery, in an operator-algebraic way that closely echoes the Dong--Harlow--Wall argument \cite{Dong:2016eik}. 
We thus obtain an exact result on the exact holographic recovery in Theorem \ref{thm:exactbest}. This new exact theorem is similar to Theorem \ref{thm:summary} but now it is in a setting that can be naturally generalized to include nonperturbative gravity corrections to give rise to Theorem \ref{thm:main}.
\begin{thm}
\label{thm:exactbest}
Let $V:\mathcal{H}_{code}\rightarrow\mathcal{H}_{phys}$ be an isometry, let $M_{code}$ be a von Neumann algebra on $\mathcal{H}_{code}$ and $M_{phys}$ be a von Neumann algebra on $\mathcal{H}_{phys}$. Let $\mathcal{E}$ be the quantum channel $V^\dagger(\cdot)V$ on $M_{phys}$. Suppose that there exists a norm one projection $\mathcal{P}_{code}:M^\prime_{phys}\rightarrow M^\prime_{code}$, and that for all normal states $\rho$ and $\varphi$ on $M^\prime_{code}$, $$S_{\rho\circ \mathcal{E}^c,\;\varphi\circ\mathcal{E}^c}(M^\prime_{phys})=S_{\rho,\varphi}(M^\prime_{code}).$$ Then $M_{code}$ is exactly correctable from $M_{phys}$.
\end{thm}
In particular, when conservation of relative entropy and the existence of a norm one projection are satisfied for both $M_{phys}$ and $M^\prime_{phys}$, the exact complementary recovery can be realized.

\section{State-independent approximate recovery in the reconstruction wedge}

With the theorems and results of Section \ref{sec:privacycorrectability} at hand, we have all the ingredients needed to derive a new result incorporating nonperturbative gravitational effects on the recovery. This requires the setting of \emph{approximate recovery} with respect to reconstruction of bulk operators in the \emph{reconstruction wedge}. This can be viewed both as a direct extension of infinite-dimensional exact recovery \cite{Gesteau:2020rtg} to a precise approximate recovery setup, that includes nonperturbative gravitational errors, and an infinite-dimensional generalization of finite-dimensional approximate recovery \cite{Hayden:2018khn}.

\subsection{Setup}
\label{sec:setup}

We consider a time slice of a CFT on the boundary of some asymptotically AdS spacetime, which we divide into two subregions $A$ and $\overline{A}$. We associate to $A$ and $\overline{A}$ von Neumann algebras $M_A$ and $M_{\overline{A}}$ acting on the CFT Hilbert space $\mathcal{H}_{phys}$, and we suppose that Haag duality \cite{Haag} is valid on the boundary:
\begin{align}
    M_{\overline{A}}=M_A^\prime.
\end{align}
We consider a bulk Hilbert space $\mathcal{H}_{code}$, as well as an isometry
\begin{align}
    V:\mathcal{H}_{code}\rightarrow\mathcal{H}_{phys}.
\end{align}
In what follows, we are interested in the quantum channel $\mathcal{E}$ given by
\begin{align} 
\begin{aligned}
\mathcal{E}:M_A&\longrightarrow \mathcal{B}(\mathcal{H}_{code}), \\
A&\longmapsto V^\dagger AV,
\end{aligned}
\end{align} 
for this isometry $V$, as well as its complementary channel $\mathcal{E}^c$ for the same isometry $V$ as
\begin{align} 
\begin{aligned}
\mathcal{E}^c:M_A^\prime&\longrightarrow \mathcal{B}(\mathcal{H}_{code}), \\
A&\longmapsto V^\dagger AV.
\end{aligned}
\end{align} 
Note that in the exact setting, these maps composed with their recovery channels would correspond to the holographic conditional expectation, and they would exactly map $M_A$ to the bulk algebra $M_a$ of its entanglement wedge, and $M_{\overline{A}}$ to the bulk algebra $M_{\overline{a}}$ of its entanglement wedge. 

In our new setup, the situation will be slightly more complicated; however, one can still associate a bulk algebra $M_a$ to the \textit{reconstruction wedge} of $A$, which corresponds to the intersection of the entanglement wedges of $A$ with respect to all \textit{mixed} states in the code subspace.\footnote{The necessity to resort to the reconstruction wedge rather than the entanglement wedge in order to obtain a state-independent reconstruction was first pointed out by Hayden and Penington in \cite{Hayden:2018khn}.} In order to take these mixed states into account (and to give them a physical interpretation), the easiest resolution is to introduce a copy of the code subspace $\mathcal{H}_{code}^{\ast}$ and consider the doubled Hilbert space $\mathcal{H}_{code}\otimes\mathcal{H}_{code}^\ast$. This doubled Hilbert space can be obtained from the GNS representation of any faithful normal (i.e. KMS) state on $\mathcal{B}(\mathcal{H}_{code})$.\footnote{For more on this technology, please see \cite{Gesteau:2020rtg}.} The maps $\mathcal{E}$ and $\mathcal{E}^c$ can then be canonically extended such that
\begin{subequations}
\begin{align}
    \mathcal{E} &\longrightarrow \mathcal{E}\ \otimes Id\quad\text{on}\quad M_A\otimes\mathcal{B}(\mathcal{H}_{code}^\ast) ,\\
    \mathcal{E}^c &\longrightarrow \mathcal{E}^c\otimes Id\quad\text{on}\quad M_A^\prime\otimes\mathcal{B}(\mathcal{H}_{code}^\ast) .
\end{align}
\end{subequations}
Furthermore, any normal mixed state on $\mathcal{B}(\mathcal{H}_{code})$ has a vector purification on $\mathcal{H}_{code}\otimes\mathcal{H}_{code}^\ast$. For technical reasons associated to the definition of the completely bounded norm and subtleties regarding tensor products of operator algebras in infinite-dimensions, in the case in which $\mathcal{H}_{code}$ is infinite-dimensional, we will choose $\mathcal{H}^\ast_{code}$ to be finite-dimensional of arbitrarily large dimension. This will not change the physics: being able to adjoin an arbitrarily large reference system to $\mathcal{H}_{code}$ will be enough to guarantee state-independent approximate recovery. In fact, this amounts to giving an arbitrarily precise approximation to an infinite-dimensional copy of $\mathcal{H}_{code}$.

\subsection{Main result}

Utilizing the setup of Section \ref{sec:setup}, we now prove our main theorem regarding approximate recovery in the reconstruction wedge for infinite-dimensional Hilbert spaces, which will be coming from the correctability and privacy correspondence. We start by applying the modified-JLMS formula to two well-chosen code states. Let $a$ denote the reconstruction wedge of $A$ and $a^\prime$ denote its complement in the bulk. Let $\mathcal{P}_{a^\prime}$ denote a normal projection of norm one from $\mathcal{B}(\mathcal{H}_{code})$ onto $M_{a^\prime}$. Before we set up further structures, it is important to note that such a projection will not exist for all possible types of von Neumann algebras. While we do not need further restrictions on the boundary von Neumann algebra, the bulk algebra cannot be completely general. The right assumption to guarantee the existence of such a projection is to suppose that
$$
M_{a^\prime}\,\text{ is \textit{purely atomic}},
$$
i.e. a product of type I factors \cite{Blackadar}. This is slightly more restrictive than just asking for $M_{a^\prime}$ to have type I, as not every direct integral of type I factors can be written as a product. However, finite-dimensional von Neumann algebras, as well as every finite or countable direct sum of type I factors, will satisfy this property. We expect this to be enough to model the bulk algebras in physically relevant situations.

Let $\rho$ be a normal state on $\mathcal{B}(\mathcal{H}_{code}\otimes\mathcal{H}_{code}^\ast)$. Then, we note that $\rho\circ (\mathcal{E}^c\otimes Id)$ is a state on $M_{\overline{A}}\otimes\mathcal{B}(\mathcal{H}_{code}^\ast)=M_A^\prime\otimes\mathcal{B}(\mathcal{H}_{code}^\ast)$.\footnote{Note that defining the tensor product of von Neumann algebras is subtle, and that there are a lot of different ways to do so. Here, we do not expand on these subtleties, as $\mathcal{H}_{code}^\ast$ is finite-dimensional, and all definitions coincide in this case.} We will apply the modified-JLMS formula for relative entropy to the two states $\rho$ and $\rho\circ (\mathcal{P}_{a^\prime}\otimes Id)$: 
\begin{align}
\begin{aligned}
|S_{\rho\circ (\mathcal{E}^c\otimes Id),\rho\circ (\mathcal{P}_{a^\prime}\otimes Id)\circ (\mathcal{E}^c\otimes Id)}(M_A^\prime\otimes\mathcal{B}(\mathcal{H}_{code}^\ast))-S_{\rho,\rho\circ (\mathcal{P}_{a^\prime}\otimes Id)}(M_{EW(\rho\circ (\mathcal{P}_{a^\prime}\otimes Id), \overline{A}\cup R)})\\
+S_{gen}(\rho,EW(\rho, \overline{A}\cup R))-S_{gen}(\rho,EW(\rho\circ (\mathcal{P}_{a^\prime}\otimes Id), \overline{A}\cup R))|&&\leq\varepsilon,
\end{aligned}
\label{eqn:setupJLMS}
\end{align}
where $\varepsilon$ is nonperturbatively small in $G_N$. Indeed, while we can naively assume that the highest order correction to this inequality is of order $G_N$, the quantum extremal surface prescription ensures that strictly non-zero nonperturbative gravitational errors persist in entanglement wedge reconstruction, with a lower bound \cite{Hayden:2018khn}
\begin{align}
    \varepsilon\sim e^{-\kappa/G_N}\,,\quad\kappa>0.
\end{align}
These nonperturbative gravitational errors share their origin with the derivation of the Page curve \cite{Penington:2019npb}: saddles with nontrivial topologies can only appear when nonperturbative effects are taken into account in the gravitational path integral. This is in contrast with the (H)RT prescription, which is only valid up to $O(1)$.

As we shall show, the second term of the left hand side of equation \eqref{eqn:setupJLMS} is zero; henceforth, this term is finite, which allows us to apply this formula. Now we analyze the second, third, and fourth terms of this equation \eqref{eqn:setupJLMS}. In the second term, $EW(\rho\circ (\mathcal{P}_{a^\prime}\otimes Id), \overline{A}\cup R)$ denotes the entanglement wedge of the boundary region $\overline{A}$ and the reference system $R$ in the bulk state $\rho\circ (\mathcal{P}_{a^\prime}\otimes Id)$. Recall that the region $a^\prime$ in the bulk is the complement of the intersection of all entanglement wedges in pure and mixed states for $A$. It follows that no matter how $\rho$ is chosen, $EW(\rho\circ (\mathcal{P}_{a^\prime}\otimes Id),\overline{A}\cup R)$ is always contained in $a^\prime\cup R$. Hence, the states $\rho$ and $\rho\circ (\mathcal{P}_{a^\prime}\otimes Id)$ coincide on it. Thus, we can conclude that the second term is zero.

The third and fourth terms correspond to the generalized entropies, i.e. the sum of the area term and the bulk correction, for the two states $\rho$ and $\rho\circ (\mathcal{P}_{a^\prime}\otimes Id)$. This may naively pose a problem as we do not know yet how to define a suitable regularization for these generalized entropies in the infinite-dimensional setting. Luckily, however, it does not matter here: as the states $\rho$ and $\rho\circ (\mathcal{P}_{a^\prime}\otimes Id)$ always coincide on the entanglement wedge of $\overline{A}\cup R$ for all states, they give rise to the same quantum extremal surface, and the third and last terms of the equation \eqref{eqn:setupJLMS} cancel out.

Incorporating these results into the equation \eqref{eqn:setupJLMS}, we obtain a much simpler identity:
\begin{align}
|S_{\rho\circ (\mathcal{E}^c\otimes Id),\rho\circ (\mathcal{P}_{a^\prime}\otimes Id)\circ (\mathcal{E}^c\otimes Id)}(M_A^\prime\otimes\mathcal{B}(\mathcal{H}_{code}^\ast))|\leq\varepsilon.
\end{align}
By Pinsker's inequality \cite[Proposition 5.23]{OhyaPetz}, this yields
\begin{align}
\label{eqn:Pinsker}
    \|\rho\circ (\mathcal{E}^c\otimes Id)-\rho\circ (\mathcal{P}_{a^\prime}\otimes Id)\circ (\mathcal{E}^c\otimes Id)\|\leq\sqrt{2\varepsilon}.
\end{align}
For convenience, we define the difference in the quantum channel $\mathcal{E}^c$ and its projected quantum channel $\mathcal{P}_{a^\prime}\circ\mathcal{E}^c$ as
\begin{align}
    \mathcal{E}^{diff}\equiv\mathcal{E}^c-\mathcal{P}_{a^\prime}\circ\mathcal{E}^c
\end{align}
such that equation \eqref{eqn:Pinsker} can be rewritten as
\begin{align}
    \label{eqn:diff}
    \|\rho\circ (\mathcal{E}^{diff}\otimes Id)\|\leq\sqrt{2\varepsilon}.
\end{align}

We now utilize the following proposition, which we come back to and prove after the main theorem: Theorem \ref{thm:main}.
\begin{prop}
\label{thm:prop}
Let $M$ and $N$ be two von Neumann algebras and let $\Phi:M\rightarrow N$ be a $\ast$-preserving normal map. Then, 
\begin{align}
\|\Phi\|_{cb}=\underset{n\in\mathbb{N}}{\mathrm{sup}}\;\underset{\rho_n}{\mathrm{sup}}\|\rho_n\circ(\Phi\otimes Id_n)\|,\end{align}
where $\rho_n$ varies over normal states on $N\otimes M_n(\mathbb{C})$.
\end{prop}
Given that $\mathcal{H}^\ast_{code}$ can be chosen to have an arbitrarily high dimension if $\mathcal{H}_{code}$ is infinite-dimensional, in that case this can be rewritten using the completely bounded norm:
\begin{align}
\|\mathcal{E}^{diff}\|_{cb}=\|\mathcal{E}^c-\mathcal{P}_{a^\prime}\circ \mathcal{E}^c\|_{cb}\leq \sqrt{2\varepsilon}.
\end{align}
This is exactly what we want! By definition of $\mathcal{P}_{a^\prime}$, the von Neumann algebra $M_a$ is private for the quantum channel $\mathcal{P}_{a^\prime}\circ \mathcal{E}^c$, so in turn, it is $\sqrt{2\varepsilon}$-private for $\mathcal{E}^c$.

As the last step, by the correctability and privacy duality as presented in Theorem \ref{thm:CKLT}, it follows that the von Neumann algebra $M_a$ is $2(2\varepsilon)^\frac{1}{4}$-correctable for the quantum channel $\mathcal{E}$. In other words, there exists a channel $\mathcal{R}:M_a\rightarrow M_A$ such that 
\begin{equation}
\|\mathcal{E}\circ \mathcal{R}-Id_{M_a}\|_{cb}\leq 2(2\varepsilon)^{\frac{1}{4}}.
\end{equation}

Note that the precise argument of Proposition \ref{thm:prop} cannot be used in the case where $\mathcal{H}_{code}$ is finite-dimensional, as $\mathcal{H}^\ast_{code}$ is only allowed to have dimension up to that of $\mathcal{H}_{code}$. However, an argument based on Jordan decomposition aids us to achieve a similar bound with a slightly less precise prefactor, which can be applied to both finite and infinite dimensional settings.

We introduce $f\in (\mathcal{B}(\mathcal{H}_{code})\otimes\mathcal{B}(\mathcal{H}_{code}^\ast))_\ast$ which is a normal linear functional of norm-$1$. We decompose this complex $f$ with real and imaginary parts such that
\begin{align}
    f=R+iI,\ R=\frac{f+f^\ast}{2},\ I=\frac{f-f^\ast}{2i},
 \end{align}
where $R$ and $I$ are self-adjoint normal linear functionals that capture real and imaginary parts of $f$ respectively. Then, their norms are bounded by the norm of $f$:
 \begin{align}
     \|R\|\leq\|f\|,\ \|I\|\leq \|f\|.
 \end{align}
 Since both $R$ and $I$ are Hermitian linear functionals, by Jordan decomposition \cite{Pedersen}, there exist unique positive linear functionals $R_+, R_-, I_+, I_-\in (\mathcal{B}(\mathcal{H}_{code})\otimes\mathcal{B}(\mathcal{H}_{code}^\ast))_\ast$ such that
\begin{align}
\begin{split}
     R=R_+-R_-\,,&\ I=I_+-I_-\,, \\
    \|R\|=\|R_+\|+\|R_-\|\,,&\ \|I\|=\|I_+\|+\|I_-\|\,.
 \end{split}
 \end{align}
 Putting these together with an insertion of $f$ to equation \eqref{eqn:diff}, we get
 \begin{align}
     \|f\circ (\mathcal{E}^{diff}\otimes Id)\| =\|(R_+-R_-+i(I_+-I_-))\circ (\mathcal{E}^{diff}\otimes Id)\|
 \end{align}
 We apply the triangle inequality:
 \begin{align}
     \begin{split}
    \|f\circ (\mathcal{E}^{diff}\otimes Id)\|
    \leq\ &\|(R_+\circ (\mathcal{E}^{diff}\otimes Id)\| +\|(R_-\circ (\mathcal{E}^{diff}\otimes Id)\| \\
     &\quad +\|(I_+\circ (\mathcal{E}^{diff}\otimes Id)\| +\|(I_-\circ (\mathcal{E}^{diff}\otimes Id)\|
     \end{split}
 \end{align}
 Given that ${R\pm}/{\|R_{\pm}\|}$ and ${I\pm}/{\|R_{\pm}\|}$ are normal states, we can apply this result to equation \eqref{eqn:diff} to conclude that 
 \begin{align}
    \|f\circ (\mathcal{E}^{diff}\otimes Id)\|\leq(\|R_+\|+\|R_-\|+\|I_+\|+\|I_-\|)\sqrt{2\varepsilon}.
 \end{align} 
 By the unity property of states for the Jordan decomposition, this simplifies as
 \begin{align}
    \|f\circ (\mathcal{E}^{diff}\otimes Id)\| \leq2\sqrt{2\varepsilon}.
\end{align}
As this holds for any normal linear functional $f$ of norm-$1$, we can deduce that
 \begin{align}
   \|\mathcal{E}^{diff}\otimes Id \|\leq2\sqrt{2\varepsilon}.
\end{align}
Indeed, we can view the composition of a quantum channel by a normal linear functional as the action of its dual quantum channel on the predual $(\mathcal{B}(\mathcal{H}_{code})\otimes\mathcal{B}(\mathcal{H}_{code}^\ast))_\ast$. The norm of the channel then corresponds to the norm of the dual channel acting on the predual, which is spanned by normal linear functionals.

We can then likewise utilize the privacy-correctability correspondence for the finite-dimensional case, and we conclude that this setup gives rise to approximate recovery. More formally, our result can be summed up and presented as Theorem \ref{thm:main} below.

\begin{thm:main}
Let $\mathcal{H}_{code}$ and $\mathcal{H}_{phys}$ be two Hilbert spaces, $V:\mathcal{H}_{code}\rightarrow \mathcal{H}_{phys}$ be an isometry, and $\mathcal{H}^\ast_{code}$ be any finite-dimensional Hilbert space of dimension smaller or equal to the one of $\mathcal{H}_{code}$. Let $M_{A}$ be a von Neumann algebra on $\mathcal{H}_{phys}$. To each normal state $\omega$ in $\mathcal{B}(\mathcal{H}_{code}\otimes\mathcal{H}^\ast_{code})_\ast$, we associate two entanglement wedge von Neumann algebras $M_{EW(\omega, A)}$ and $M_{EW(\omega, \overline{A}\cup R)}$ of operators on $\mathcal{H}_{code}\otimes\mathcal{H}_{code}^\ast$, such that $M_{EW(\omega, A)}\subset\mathcal{B}(\mathcal{H}_{code})\otimes Id$ and $M_{EW(\omega, \overline{A}\cup R)}\subset M_{EW(\omega, A)}^\prime$. Let
\begin{align*}
    M_{a}:=\underset{\omega}{\bigcap}\;M_{EW(\omega, A)}
\end{align*} 
be the reconstruction wedge von Neumann algebra on $\mathcal{H}_{code}$, and suppose that $M_{a^\prime}$, the relative commutant of $M_a$ in $\mathcal{B}(\mathcal{H}_{code})\otimes Id$, is a product of type $I$ factors. Suppose that for all choices of $\mathcal{H}^\ast_{code}$ and all pairs of states $\rho,\omega$ in $\mathcal{B}(\mathcal{H}_{code}\otimes\mathcal{H}^\ast_{code})_\ast$ such that $S_{\rho,\omega}(M_{EW(\omega, \overline{A}\cup R)})$ is finite, we have the following modified-JLMS condition: 
\begin{align*}
|S_{\rho\circ (\mathcal{E}^c\otimes Id),\omega\circ (\mathcal{E}^c\otimes Id)}(M_A^\prime\otimes\mathcal{B}(\mathcal{H}_{code}^\ast))-S_{\rho,\omega}(M_{EW(\omega, \overline{A}\cup R)})&\\
+S_{gen}(\rho,EW(\rho, \overline{A}\cup R))-S_{gen}(\rho,EW(\omega, \overline{A}\cup R))|&\ \leq\ \varepsilon,
\end{align*}
where $\mathcal{E}$ and $\mathcal{E}^c$ refer to the respective restrictions of $A\mapsto V^\dagger A V$ to $M_A$ and $M^\prime_A$, and the function $S_{gen}(\rho,EW(\omega, \overline{A}\cup R))$ depends only on the restrictions of $\rho$ and $\omega$ to $M_{a^\prime}\otimes\mathcal{B}(\mathcal{H}_{code}^\ast)$.\footnote{We note that this is more relaxed than restriction conditions to any specific entanglement wedge, i.e. $M_{EW(\omega, \overline{A}\cup R)}$.}
Then, there exists a quantum channel $\mathcal{R}:M_{a}\rightarrow M_{A}$ such that
\begin{align*}
    \|\mathcal{E}\circ\mathcal{R}-Id_{M_{a}}\|_{cb} \leq 2(2\varepsilon)^\frac{1}{4}.
\end{align*}
For finite-dimensional $\mathcal{H}_{code}$, we find a bound $2\sqrt{2\sqrt{2\varepsilon}}$ instead.
\end{thm:main}

In Theorem \ref{thm:main}, the von Neumann algebra $M_a$ corresponds to the observables in the reconstruction wedge of the boundary region $A$. Depending on this boundary region $A$, we defined the reconstruction wedge, as the notion required for nonperturbative gravitational effects, to be the intersection of all entanglement wedges of pure and mixed states on $A$. Then, what this theorem entails physically is that utilizing this new reconstruction wedge we can guarantee state-independent reconstruction. On the contrary, observables that are outside the reconstruction wedge but inside a specific entanglement wedge, cannot be reconstructed on this region $A$ in a state-independent manner. 

There often exists a macroscopic difference between a fixed entanglement wedge in a code subspace and the reconstruction wedge associated to this code subspace, due to the bulk term in the quantum extremal surface formula that relates the boundary entanglement entropy to the bulk generalized entropy of the entanglement wedge. This bulk term can sometimes be dominant in the presence of a large number of gravitons. This shows the importance of jumps of the entanglement wedge in the presence of gravity; it further suggests that the reconstruction wedge is an important object in the holographic dictionary. 

We want to emphasize that this remains true in a generic Hilbert space formalism, both finite and infinite dimensional cases, for the bulk and boundary construction. As the reconstruction wedge is a natural object to consider in the context of quantum extremal surfaces, it would be interesting to determine the boundary dual of the area of the reconstruction wedge.

We now prove Proposition \ref{thm:prop}, which is essential to complete the proof of Theorem \ref{thm:main}.\footnote{We thank Vern Paulsen for communicating ideas of this proof.} We first recall Proposition \ref{thm:prop}.
\begin{thm:prop}
Let $M$ and $N$ be two von Neumann algebras and let $\Phi$ be a $\ast$-preserving normal map. Then, 
\begin{align}
\|\Phi\|_{cb}=\underset{n\in\mathbb{N}}{\mathrm{sup}}\;\underset{\rho_n}{\mathrm{sup}}\|\rho_n\circ(\Phi\otimes Id_n)\|,\end{align}
where $\rho_n$ varies over normal states on $N\otimes M_n(\mathbb{C})$.
\end{thm:prop}
We first prove the two following lemmas that lead to proving this propostition.
\begin{lem}
\label{lem:sa}
Let $M$ be a von Neumann algebra and $\omega$ be a normal $\ast$-preserving linear functional on $M$. Then, $\omega$ attains its norm on a self-adjoint element.
\end{lem}

\begin{proof}
Let $A\in M$ and $\xi$ a complex number of modulus 1 such that 
\begin{align}|\omega(A)|=\xi\omega(A).\end{align}
Now let \begin{align}H:=\frac{\xi A+\bar{\xi} A^\dagger}{2}.\end{align}
By the triangle inequality, $H$ is Hermitian of norm smaller or equal to that of $A$. Moreover, as $\omega$ is $\ast$-preserving, \begin{align}\omega(H)=|\omega(A)|.\end{align} This concludes the proof that $\omega$ attains its norm on self-adjoint elements.
\end{proof}

\begin{lem}
\label{lem:satocb}
Let $\Phi:M\longrightarrow N$ be a normal map between von Neumann algebras. Then, \begin{align}\|\Phi\|_{cb}=\underset{n\in\mathbb{N}}{\mathrm{sup}}\;\|\Phi\otimes Id_n\|_{sa},\end{align}where $\|\cdot\|_{sa}$ means that the supremum in the definition of the norm is restricted to self-adjoint elements of $M\otimes\mathcal{M}_n(\mathbb{C})$.
\end{lem}

\begin{proof}
Let $A\in M$, and define \begin{align}B:=\begin{pmatrix}0 & A\\ A^\dagger & 0 \end{pmatrix}\in M\otimes M_2(\mathbb{C}).\end{align}$B$ is self-adjoint, and \begin{align}\|B\|=\|A\|.\end{align}Moreover, \begin{align}\|(\Phi\otimes Id_2)(B)\|=\mathrm{max}(\|\Phi(A)\|,\|\Phi(A^\dagger)\|).\end{align} This last equality shows that we have the upper and lower bounds \begin{align}\|\Phi\|\leq\|\Phi\otimes Id_2\|_{sa}\leq\|\Phi\otimes Id_2\|.\end{align} Iterating the tensor product with $Id_2$ and taking the supremum, we get \begin{align}\|\Phi\|_{cb}=\underset{n\in\mathbb{N}}{\mathrm{sup}}\;\|\Phi\otimes Id_n\|_{sa}.\end{align}\end{proof}

With these two lemmas in hand, we are ready to prove Proposition \ref{thm:prop}.  By Lemma \ref{lem:satocb}, for $n\in\mathbb{N}$, we have that 
\begin{align}
    \|\Phi\|_{cb}=\underset{n\in\mathbb{N}}{\mathrm{sup}}\underset{A_n\;\text{self-adjoint},\;\|A_n\|=1}{\mathrm{sup}}\|(\Phi\otimes Id_n)(A_n)\|.
\end{align} 
We note that the $(\Phi\otimes Id_n)(A_n)$ are all self-adjoint (as $\Phi$ is $\ast$-preserving), and that the norm of a self-adjoint operator in a von Neumann algebra can be obtained by taking the supremum of its values against normal states. Hence, we have  
\begin{align}
\|\Phi\|_{cb}=\underset{n\in\mathbb{N}}{\mathrm{sup}}\underset{A_n\;\text{self-adjoint},\;\|A_n\|=1}{\mathrm{sup}}\underset{\rho_n}{\mathrm{sup}}\;|\rho_n((\Phi\otimes Id_n)(A_n))|.
\end{align} 
This precisely yields 
\begin{align}
\|\Phi\|_{cb}=\underset{n\in\mathbb{N}}{\mathrm{sup}}\;\underset{\rho_n}{\mathrm{sup}}\;\|\rho_n\circ(\Phi\otimes Id_n)\|_{sa}.
\end{align}
We now use Lemma \ref{lem:sa} to drop the self-adjoint condition on the right hand side:
\begin{align}
\|\Phi\|_{cb}=\underset{n\in\mathbb{N}}{\mathrm{sup}}\;\underset{\rho_n}{\mathrm{sup}}\;\|\rho_n\circ(\Phi\otimes Id_n)\|,
\end{align}
which concludes the proof of Proposition \ref{thm:prop}.

\section{State-dependent recovery beyond the reconstruction wedge}

The reconstruction result we derived thus far is completely state-independent: operators inside the reconstruction wedge are inside the entanglement wedge of any pure or mixed state, and we saw that infinite-dimensional privacy/correctability duality allowed to derive state-independent bounds on correctability from a modified version of the JLMS formula. 

However, some observables that reach deeper than the reconstruction wedge (for example, the black hole interior), cannot be reconstructed in such a state-independent fashion \cite{Papadodimas:2012aq,Papadodimas:2013jku}. Instead, their boundary representatives are \textit{state-dependent}, and this feature has been at the center of many discussions, including understandings of finite \cite{Hayden:2018khn} to infinite-dimensional \cite{Gesteau:2020rtg} Hilbert space settings. In particular, it provides a lot of tools for the resolution of the information paradox \cite{Penington:2019npb}. 

\subsection{$\alpha$-bits}

Discussions on state-dependence are rooted in the works of Papadodimas--Raju \cite{Papadodimas:2012aq,Papadodimas:2013jku}, but here we will focus on the modern approach to the problem by Hayden--Penington, through the notion of $\alpha$-bits \cite{Hayden:2018khn}. The main idea is that only subspaces whose dimension grows like a certain power of the dimension of the whole code subspace can be recovered inside a black hole. 

We first recall the setup of \cite{Hayden:2018khn}, and show how the arguments can be generalized to infinite-dimensions. The idea is to consider a finite-dimensional code subspace, that factorizes into a black hole and an exterior piece 
\begin{align}
    \mathcal{H}=\mathcal{H}_{BH}\otimes\mathcal{H}_{ext},
\end{align}
with the area-dependence given by
\begin{align}
    \underset{G_N\rightarrow 0}{\mathrm{lim}}G_N\,\mathrm{log}\,d_{BH}=\frac{\mathcal{A}_0}{4},
\end{align}
where $\mathcal{A}_0$ is the black hole area. Then, two bulk regions are considered, both anchored in the same boundary region $A$, of respective areas $\mathcal{A}_1$ and $\mathcal{A}_2$. The region $\mathcal{A}_1$ corresponds to the entanglement wedge of $A$ for a pure boundary state and contains the black hole, whereas the region $\mathcal{A}_2$ corresponds to the entanglement wedge of a thermal boundary state and does not contain the black hole. We then assume that 
\begin{align}
    \alpha:=\frac{\mathcal{A}_2-\mathcal{A}_1}{\mathcal{A}_0}
\end{align}
is strictly smaller than 1. We also assume that the geometric and matter contributions are small outside of the black hole, then it is argued in \cite{Hayden:2018khn} that there are only 2 possible quantum extremal surfaces for the entanglement wedge, of respective areas $\mathcal{A}_1$ and $\mathcal{A}_2$. In order to be able to get a state-independent recovery for the black hole, one then needs the second quantum extremal surface to always be dominant: this is the case when the dimension of the subspace we wish to decode satisfies 
\begin{align}
    d\leq e^{\alpha \frac{\mathcal{A}_0}{4G_N}},
\end{align}
which is strictly less than the dimension of the black hole Hilbert space. Hence, the previous reasoning only shows that the black hole interior reconstruction is only possible in a state-dependent way. It turns out this $d\leq e^{\alpha \frac{\mathcal{A}_0}{4G_N}}$ is the best possible bound for the dimension of a subspace whose reconstruction can be achieved in a state-independent manner \cite{Hayden:2018khn}. 

Let us generalize this argument to the case of an infinite-dimensional boundary Hilbert space. Suppose that the black hole Hilbert space is entangled with a Hilbert space $\mathcal{H}_r$ of dimension $d^\prime$ given by
\begin{align}
    d'=e^{\alpha^\prime \frac{\mathcal{A}_0}{4G_N}},\quad\text{for}\quad\alpha^\prime>\alpha .
\end{align}
This Hilbert space $\mathcal{H}_r$ corresponds to a reference system that takes into account all the degrees of freedom the black hole can be entangled to - for example, Hawking radiation.
We can then, by our previous reasoning, construct a subspace 
\begin{align}
    \mathcal{H}_S\subset\mathcal{H}_{code}\otimes\mathcal{H}_{r}
\end{align}
such that $(M_{a^\prime}\otimes\mathcal{B}(\mathcal{H}_r))\cap\mathcal{B}(\mathcal{H}_S)$ can be reconstructed from $M_A^\prime\otimes\mathcal{B}(\mathcal{H}_r)$ up to nonperturbative error. Indeed, perturbing outside of the black hole but still between the two candidate quantum extremal surfaces cannot change the quantum extremal surface, so it suffices to perturb in such a way around a state which is maximally entangled between the black hole and the reference system in order to construct a subspace of states whose entanglement wedge contains the black hole. Note that if the dimension of $\mathcal{H}_S$ is chosen to be small enough, this is still true even for states that are entangled with a second reference system whose dimension equals the one of $\mathcal{H}_S$. In other words, following the proof of Theorem 5.1, $(M_{a^\prime}\otimes \mathcal{B}(\mathcal{H}_r))\cap\mathcal{B}(\mathcal{H}_S)$ is $\delta$-correctable for $M_A^\prime\otimes\mathcal{B}(\mathcal{H}_r)$ for some nonperturbatively small $\delta>0$, and also $\delta^\prime$-private for $M_A\otimes Id$ for some nonperturbatively small $\delta^\prime>0$ by privacy/correctability duality. This shows the impossibility to recover the black hole interior for any $\alpha^\prime>\alpha$.

\subsection{Universal recovery channel}
\label{sec:recoverychannel}

It is proven in a finite-dimensional context that a universal recovery channel, known as the \textit{twirled Petz map}, can be used to recover any state in a subspace for which the JLMS formula is satisfied \cite{Cotler:2017erl}. This statement was later extended to the infinite-dimensional case by Faulkner, Hollands, Swingle and Wang \cite{Faulkner:2020iou,Faulkner:2020kit}. We use their ideas and rephrase it in the following form for our purposes.

\begin{thm}
\label{thm:universalcorrect} 
Let $\mathcal{E}:M\longrightarrow N$ be a unital quantum channel between two von Neumann algebras $M$ and $N$. Let $\rho,\omega$ be two normal states on $N$ such that $\omega$ is faithful. Then, there exists a channel $\alpha: N\longrightarrow M$, that just depends on $\omega$ and $\mathcal{E}$ such that
\begin{align}
S_{\rho,\omega}(N)-S_{\rho\circ \mathcal{E},\omega\circ \mathcal{E}}(M)\geq \frac{1}{4} \|\rho\circ \mathcal{E} \circ \alpha - \rho\|^2.
\end{align}
\end{thm}

We first explain how this Theorem \ref{thm:universalcorrect} is relevant to our context, without going into details. For our setup, the von Neumann algebras $M$ and $N$ act respectively on the boundary Hilbert space and on the code subspace, and $\mathcal{E}$ represents a boundary-to-bulk map for operators. Note that if the left hand side of this inequality is small, then it means that the twirled Petz map $\alpha$ is a good choice of recovery channel. This is true when the JLMS condition is satisfied; we emphasize a crucial subtlety that this is different from the modified-JLMS condition. Hence, the twirled Petz map will be a good choice of recovery map only for states that have the same entanglement wedge, for example when the entanglement wedge doesn't jump and coincides with the reconstruction wedge, or in $\alpha$-bit spaces.

Hence, theorem \ref{thm:universalcorrect} shows that we may be able to use the twirled Petz map to decode bulk operators in AdS/CFT, both for the reconstruction wedge (with an appropriate choice of $\omega$), and for $\alpha$-bit subspaces of the code subspace. However, it is important to emphasize that while the privacy/correctability duality provides a robust proof of bulk reconstruction with nonperturbatively small error in $G_N$ in the reconstruction wedge, it is still unclear whether an argument based on the twirled Petz map will be exact at all orders in perturbation theory \cite{Hayden:2018khn}. Moreover, the bound we get is only in terms of the operator norm, and not in terms of the completely bounded norm. Thus, this type of approximate reconstruction is much weaker than the one we derived through privacy/correctability duality.

\section{Discussion}
\label{sec:discussion}

We developed in this paper an operator algebraic framework that includes nonperturbative gravity corrections to entanglement wedge reconstruction. We described it by first formulating the \emph{exact} entanglement wedge reconstruction via Theorem \ref{thm:summary} that summarizes known results, and investigated the conditional expectation structure of the exact holographic map to form the \emph{approximate} setting of entanglement wedge reconstruction.

While the exact setting has multiple shortcomings due to gravitational aspects, the approximate nature of bulk reconstruction provides vast important implications. It is therefore crucial to generalize the infinite-dimensional picture \cite{Kang:2019dfi,Gesteau:2020rtg,Faulkner:2020hzi} to the approximate setting. We achieve this as portrayed in Theorem \ref{thm:main}, which shows that under a modified-JLMS assumption for the code subspace tensored with a reference system, the local algebra of the \textit{reconstruction wedge} of a boundary region can be approximately reconstructed in a state-independent and dimension-independent manner, as long as its bulk complement is a product of type $I$ factors. In particular, this is the case for finite-dimensional code subspaces.  

With this framework, we further incorporated beyond the state-independent settings, which requires going beyond the reconstruction wedge. As shown in \cite{Hayden:2018khn}, reconstruction beyond the reconstruction wedge may only be state-dependent, as only the $\alpha$-bits of the black hole can be reconstructed on the boundary. We showed that our formalism is still enough to account for state-dependent reconstructions with an infinite-dimensional boundary, by resorting to privacy/correctability duality and our main theorem. Furthermore, we pointed out that the universal recovery channel known as the `twirled Petz map' achieves approximate recovery in both the reconstruction wedge and $\alpha$-bit cases, whenever recovery is possible, but only at the order $G_N^0$ and in the sense of the operator norm rather than the completely bounded norm.

On top of formulating new approximation results both in the state-dependent and in the state-independent case for an arbitrary boundary Hilbert space, this paper can also be seen as containing most known results on exact and approximate entanglement wedge reconstruction as particular cases. Here, we intended to present their most general formulation, which would be useful in the case of an actual CFT on the boundary. This allowed us to construct an exhaustive dictionary between physical notions relevant to holography and operator-algebraic concepts.

We expect our results to be physically relevant at least whenever the boundary theory can be reasonably well-described in terms of algebraic quantum field theory (AQFT). This approach, based on the Haag--Kastler axioms \cite{HaagKastler}, has several shortcomings, but has proven to be quite useful, in particular for modelling 2d Lorentzian conformal field theories, thanks to the framework of conformal nets \cite{KawahigashiLongo}. In examples of AdS${_3}$/CFT${_2}$ where the boundary is described by conformal nets, our results should in particular apply. We also expect that in higher-dimensional examples of holography, at least some of the structures and basic mechanisms uncovered here should survive. We hope to make these statements more precise in the future.

Another interesting place to test out this operator algebraic formulation would be in (limits of) random tensor networks, where error-correction becomes approximate. This, however, would require extending our theorem to the case of non-isometric maps, which we leave to future research.

We emphasize that proving entanglement wedge reconstruction results in AdS/CFT within this operator-algebraic framework would increase the level of rigor of current derivations, which rely on the Euclidean gravitational path integral. In low dimensions, the Euclidean path integral is surprisingly powerful in capturing a lot of the UV physics of the CFT, and proving results such as the quantum extremal surface formula or the island prescription. However, this path integral is thought to capture the physics of an average of gravitational theories, rather than a single one. Hence, if we want to rigorously show entanglement wedge reconstruction for a single boundary theory, we will need to tackle the problem directly within this operator-algebraic setup. Moreover, in dimension 4 or higher, it is unclear whether the gravitational path integral will be a good enough tool to capture the effects described in this paper. Even more so in that context, it is a very important problem to understand the relation between the semiclassical effective field theory and the UV complete one directly at the level of the operator algebras describing the CFT.

Further, we hope to draw a precise link between our operator-algebraic framework and geometry, and to maybe be able to derive the bulk Einstein equation in the spirit of \cite{Lewkowycz:2018sgn}. In particular, recent work \cite{Bousso:2020yxi} has shown that the Connes cocycle, a purely operator-algebraic quantity on the boundary, has a precise geometric interpretation in the bulk as a kink transform. Given that in the exact framework, the Connes cocycle is preserved under the holographic conditional expectation, it would be interesting to study it in our approximate setting, and try to formulate a more precise correspondence between operator algebraic quantities on the boundary and geometric data in the bulk. In particular, this could allow one to understand the bulk dynamics and the bulk Lorentzian structure.

Another interesting research direction is to try to study in detail how some apparent paradoxes raised by the framework of exact entanglement wedge reconstruction are resolved by going to the approximate case. This was done in finite-dimensions by Penington \cite{Penington:2019npb} in the case of the information paradox, but for instance, the tension between the lack of additivity of entanglement wedges and the boundary Reeh-Schlieder theorem, first uncovered by Kelly \cite{Kelly:2016edc} and precisely stated in the operator-algebraic context by Faulkner \cite{Faulkner:2020hzi}, still seems quite mysterious, and can only be treated in an infinite-dimensional setting. Ultimately, we hope that our approximate statements for infinite-dimensional boundary Hilbert spaces will be a first step towards a fully consistent formulation of the emergence of the bulk from the entanglement structure of the CFT.

We now conclude our analysis by giving a more detailed interpretation of our modified-JLMS condition, and explaining how our results are related to the recent progress on the information paradox.

\subsection{Generalized entropies for quantum extremal islands}

Considering two states in the same code subspace, we know that their entanglement wedges may be different \cite{Dong:2016eik}. In other words, there is a difference in the term regarding generalized entropy which effectively imposes changes in the original JLMS formula, that corresponds to the reconstruction of a single entanglement wedge. To take these contributions into account, we consider instead the modified-JLMS condition. 

For the modified-JLMS formula, not only did we have to assume the relative entropy conservation between the bulk and the boundary, but we also needed the formula to hold for the states supported on the reference system $\mathcal{H}_{code}^\ast$. In fact, this is because we had to define the map $\mathcal{E}^c\otimes Id$ in our proof. We use the privacy and correctability duality theorem in our proof which can only hold for the completely bounded norm, whose use requires adding an extra reference system. 

Due to this setup with an extra reference system, we can apply our framework to various different physical settings. The entanglement wedge of $\mathcal{H}_{code}^\ast$ can be nontrivial and significantly change our description of bulk reconstruction. In some cases, though not always, this auxiliary system can be identified with the space of semiclassical states of the Hawking radiation or of another boundary.

The idea that black hole interior degrees of freedom can be encoded in an auxiliary system traces back to the identification between wormholes and thermofield double states \cite{Maldacena:2013xja}, which is a well suited model to think about holographic entanglement. In this setup with the auxiliary system, the black hole interior simply becomes the other side of a Lorentzian wormhole. It follows that the entanglement wedge of the second boundary corresponds to what would be the black hole interior from the viewpoint of the first boundary. In particular, the modified-JLMS formula will be valid for states, in the code subspace, supported on both boundaries. 

The necessity to consider the reconstruction wedge and the state-dependence only appears when nonperturbative gravitational corrections to entanglement wedge reconstruction are taken into account. These nonperturbative corrections, although individually insignificant, may pile up in the presence of a nontrivial amount of gravitational corrections until they significantly alter the encoding of bulk information. This is, for example, what happens for black hole interior reconstruction from the Hawking radiation after the Page time. Taking into account a conformal field theory bath, it is possible to have nontrivial nonperturbative effects from replica wormholes to form quantum extremal islands for the semiclassical bulk. Then, the interior of an island is encoded in the early Hawking radiation of the black hole and the entropy of the radiation is computed from the quantum island formula, which generalizes the quantum extremal surface prescription. Applying this formula leads to a coherent semiclassical derivation of the Page curve.

Replica wormhole calculations \cite{Almheiri:2019qdq, Penington:2019kki} considerably strengthen the evidence for the validity of such an island formula. It would be very interesting to derive a justification for it directly within the boundary theory, without resorting to the gravitational path integral.

\subsection{Application to the information paradox}

In this paper, we explained how to formulate approximate entanglement wedge reconstruction for an infinite-dimensional boundary Hilbert space in the case where the boundary is divided into two regions $A$ and $\overline{A}$. We obtained that in the case where there is a black hole in the bulk, the quantum extremal surface associated to the boundary region $A$ experiences jumps between pure and mixed states. This allows a state-dependent reconstruction of the interior. 

As outlined in \cite{Penington:2019npb}, one can perform a similar analysis in the case of an evaporating black hole, whose semiclassical description is encoded in part in the CFT and in part in a bath of Hawking radiation collected during evaporation. In order to do this, one must be able to collect Hawking radiation by imposing transparent boundary conditions on the AdS boundary. 

In this case, the whole boundary plays the role of the boundary region $A$, while the Hawking radiation plays the role of $\overline{A}$. It can be shown that after the Page time, the radiation bath contains enough information so that a quantum extremal surface that contains most of the black hole interior appears, allowing the interior to be inside the entanglement wedge of the bath. However, the reconstruction of interior operators is state-dependent for a sufficiently later time after the Page time, as the quantum extremal surface disappears when the bulk system is in an overall mixed state. It would be interesting to apply our framework to the setting where the UV-complete dual contains a bath of Hawking radiation in addition to the CFT. We expect the difficulty comes from the fact that the bulk-to-boundary map is only approximately an isometry after the Page time.

\subsection{The gravitational path integral and quantum error-correction}

The most convincing argument for the Page curve in holographic theories in the literature has been derived through the use of the Euclidean gravitational path integral, by including contributions coming from nontrivial topologies \cite{Almheiri:2019qdq,Penington:2019kki}. Even for the regular quantum extremal surface formula, the Euclidean path integral is our best justification \cite{Dong:2017xht}. However, it is still not completely understood why the Euclidean path integral knows so much about the UV degrees of freedom of the theory. Recent calculations in two and three dimensions have put forward a possible interpretation, which can be made rigorous in the case of Jackiw--Teitelboim gravity \cite{Saad:2019lba} or Narain CFTs \cite{Maloney:2020nni}: the Euclidean path integral would only allow to calculate \emph{averaged} quantities over a moduli space of dual theories, rather than quantities in a single unitary theory. For example, this interpretation may provide a resolution to the factorization problem.

However, it is still very important to understand how semiclassical calculations work in a single unitary holographic theory. Indeed, if we are to describe a non-self-averaging observable in the bulk, it will be sensitive to the specific features of individual theories that compose the ensemble of boundary theories; importantly, calculating the gravitational path integral will not likely suffice. Moreover, in four or higher dimensions, it is unclear whether the gravitational path integral will be as successful, even if it is possible to think of an ensemble of boundary theories at all. For instance, $\mathcal{N}=4$ super Yang--Mills is thought to be a unique boundary theory.

The operator algebraic approach does not a priori rely on any kind of averaging procedure or any explicit path integral formula, but on a bulk-to-boundary map between the code subspace and the boundary Hilbert space. In that sense, it allows for an understanding of the semiclassical limit of a single unitary holographic theory. With our approximation theorem in hand, this approach can also handle the nonperturbative corrections that appear through nontrivial topologies in the gravitational path integral. In particular, we showed that even in infinite dimensions, which is the relevant setting for quantum field theories, approximate recovery is possible inside the reconstruction wedge. For local operators that are in the entanglement wedge of a given state, but not contained in all entanglement wedges of pure and mixed states, reconstruction can only be performed in a state-dependent manner. The fact that this reconstruction wedge is strictly smaller than some entanglement wedges shows that gravity affects some large region of the bulk that we tackle in our operator-algebraic framework.

The relationship between the gravitational path integral and ensembles has recently been investigated in detail, and given a new interpretation through the notion of Hilbert space of closed universes. This is relevant to considering the bulk with multi-boundaries. This Hilbert space describes possible nucleations of baby universes in the bulk, which correspond to new asymptotic boundaries \cite{Marolf:2020xie,Gesteau:2020wrk}. In this approach, the Euclidean path integral becomes dependent on the state of these baby universes, which are described by a commutative algebra of observables at infinity. In the Hartle--Hawking state of baby universes, the Euclidean path integral is dual to an ensemble of boundary theories, while in some very specific baby universe states, known as $\alpha$-states, which correspond to the basic superselection sectors of the theory, it actually computes observables in a single member of the dual ensemble. 

It is then interesting to ask whether this distinction between $\alpha$-states and ensemble averaging can be understood within the operator-algebraic framework of quantum error-correction, and more generally, how the notion of ensemble averaging can be encoded within our approach. Such a result would clarify the link between the Euclidean path integral and unitary theories, and explain why the Euclidean path integral works so well to understand some subtle features of gravity. Moreover, it should shed light on the Lorentzian structure of the bulk, as the operator-algebraic framework can be well-adapted to Lorentzian signature. We hope to return to these questions in future work.

\section*{Acknowledgments}
The authors are grateful to Juan Felipe Ariza Mejia, Daniel Harlow, Matilde Marcolli, and Vern Paulsen for discussions.
M.J.K. is supported by a Sherman Fairchild Postdoctoral Fellowship, the U.S. Department of Energy, Office of Science, Office of High Energy Physics, under award number DE-SC0011632, and the National Research Foundation of Korea (NRF) grants NRF-2020R1C1C1007591 and NRF-2020R1A4A3079707.
E.G. would like to thank Matilde Marcolli for her guidance and constant support.

\bibliography{monica}

\end{document}